\documentclass[a4paper,12pt]{article}

\usepackage{amsmath,amssymb,amsthm}
\usepackage[english]{babel}
\usepackage{hyperref}
\usepackage[all]{xy}

\theoremstyle{plain}
\newtheorem{theorem}{Theorem}
\newtheorem{proposition}[theorem]{Proposition}
\newtheorem{corollary}[theorem]{Corollary}
\newtheorem{lemma}[theorem]{Lemma}

\theoremstyle{definition}
\newtheorem{definition}[theorem]{Definition}
\newtheorem{example}[theorem]{Example}
\newtheorem{remark}[theorem]{Remark}

\newcommand{\ff }{\mathbb{F}}

\newcommand{\fq }{\mathbb{F}_q}

\newcommand{\fqm }{\mathbb{F}_{q^m}}

\newcommand{\ba }{{\bf a}}
\newcommand{\bb }{{\bf b}}
\newcommand{\bc }{{\bf c}}
\newcommand{\bd }{{\bf d}}

\newcommand{\bg }{{\bf g}}
\newcommand{\bh }{{\bf h}}

\newcommand{\bx }{{\bf x}}
\newcommand{\by }{{\bf y}}

\newcommand{\wt}{\mathrm{wt}}
\newcommand{\supp}{\mathrm{supp}}
\newcommand{\Rsupp}{\mathrm{Rsupp}}
\newcommand{\rk}{\mathrm{rk}}
\newcommand{\Tr}{\mathrm{Tr}}

\title{On defining generalized rank weights}
\author{Relinde Jurrius\thanks{\texttt{relinde.jurrius@unine.ch}}
\and Ruud Pellikaan\thanks{\texttt{g.r.pellikaan@tue.nl}}}

\begin{document}

\maketitle

\begin{abstract}
This paper investigates the generalized rank weights, with a definition implied by the study of the generalized rank weight enumerator. We study rank metric codes over $L$, where $L$ is a finite extension of a field $K$. This is a generalization of the case where $K=\fq $ and $L=\fqm$ of Gabidulin codes to arbitrary characteristic. We show equivalence to previous definitions, in particular the ones by Kurihara-Matsumoto-Uyematsu \cite{kurihara:2012,kurihara:2013},  Oggier-Sboui \cite{oggier:2012} and Ducoat \cite{ducoat:2013}. As an application of the notion of generalized rank weights, we discuss codes that are degenerate with respect to the rank metric.
\end{abstract}

\section{Introduction}

Error-correcting codes with the rank distance were introduced by Gabidulin \cite{gabidulin:1985}. Recently they have gained a lot of interest because of their application to network coding. In network coding, messages are not transmitted over a single channel, but over a network of channels. This application induced a lot of theoretical research to rank metric codes. \\
Many notions in the theory for codes with the Hamming metric have an equivalent notion for codes with the rank metric. We studied the rank-metric equivalent of the weight enumerator and several generalizations of it \cite{jurrius:2014}. From this theory, a definition of the generalized rank weights follows. These are the rank metric equivalence of the generalized Hamming weights. Definitions of the generalized rank weights were already proposed: we show here that the definition that follows from our work leads to the same values as all of the proposed definitions. In particular, this means all previously proposed definitions are equivalent for rank metric codes over finite fields.  \\
This paper investigates the generalized rank weights of a code over $L$, where $L$ is a finite Galois extension of a field $K$. This is a generalization of the case where $K=\fq $ and $L=\fqm$ of Gabidulin codes \cite{gabidulin:1985} to arbitrary characteristic as considered by Augot-Loidreau-Robert \cite{augot:2013,augot:2014}. \\
As a small application of the generalized rank weights, we discuss the concept of degenerate codes.

\section{Rank metric codes and weights}

Let $K$ be a field and let $L$ be a finite Galois extension of $K$. A \emph{rank metric code} is an $L$-linear subspace of $L^n$. To all codewords we associate a matrix as follows. Choose a basis $B=\{ \alpha_1, \ldots ,\alpha _m \} $ of $L$ as a vector space over $K$. Let $\bc =(c_1, \ldots ,c_n)\in L^n$. The $m \times n$ matrix $M_B(\bc)$ is associated to $\bc $ where the $j$-the column of $M_B(\bc )$ consists of the coordinates of $c_j$  with respect to the chosen basis: $c_j = \sum_{i=1}^m c_{ij}\alpha_i$. So $M_B(\bc )$ has entries $c_{ij}$. \\
The $K$-linear row space in $K^n$ and the rank of $M_B(\bc )$ do not depend on the choice of the basis $B$,
since for another basis $B'$ there exists an invertible matrix $A$ such that $M_B(\bc)= AM_{B'}(\bc)$.
The rank weight $\wt_R(\bc )=\rk(\bc)$ of $\bc $ is by definition the rank of the matrix $M_B(\bc )$,
or equivalently the dimension over $K$ of the row space of $M_B(\bc )$. This definition follows from the rank distance, that is defined by $d_R(\bx, \by ) = \rk(\bx-\by)$. The rank distance is in fact a metric on the collection of all $m \times n$ matrices, see \cite{gabidulin:1985,augot:2013}.

The following is from \cite[Definition 1]{jurrius:2014}.
\begin{definition}
Let $C$ be an $L$-linear code.
Let $\bc \in C$. Then $\Rsupp(\bc )$, the {\em rank support} of $\bc$
is the $K$-linear row space of $M_B(\bc )$. So $\wt_R(\bc)$ is the dimension of $\Rsupp(\bc)$.  Let $D$ be an $L$-linear subcode of $C$.
Then $\Rsupp(D)$, the {\em rank support} of $D$ is the $K$-linear space generated by the $\Rsupp(\bd )$ for all $\bd \in D$.
Then $\wt_R(D)$, the {\em rank support weight} of $D$ is the dimension of $\Rsupp(D)$.
\end{definition}

Note that this definition is the rank metric case of the support weights, or weights of subcodes,
of codes over the Hamming metric.
An \emph{isometry} with respect to the rank metric is an automorphism of $L^n$ that preserves the rank distance.
The following theorem characterizes when we call two rank metric codes ``equivalent'':

\begin{theorem}
The group $\mbox{Iso}(L^n)$ of isometries of $L^n$ with respect to the rank distance
is generated by the scalar multiplications $\lambda :L^n \rightarrow L^n$ with $\lambda$ in $L^*$
given by $\lambda \bx = (\lambda x_1, \ldots , \lambda x_n)$,
and the general linear group $\mbox{GL}(n; K)$ of invertible $n\times n$ matrices with entries in $K$.
The group $\mbox{Iso}(L^n)$  is isomorphic to the product group $L^*/K^* \times \mbox{GL}(n; K)$
\end{theorem}

\begin{proof}
The proof in \cite{berger:2002} and \cite[Theorem 1]{berger:2003} in case of finite fields is easily generalized.
\end{proof}

We prove some basic properties of the rank weight and rank support.

\begin{proposition}\label{p-Rsupp}
Let $C$ be an $L$-linear code.\\
$(1)$ Let $\bx \in L^n$ and $\alpha \in L^*$. Then $\Rsupp(\alpha \bx)= \Rsupp(\bx)$.\\
$(2)$ Let $\bx, \by \in L^n$. Then $\Rsupp(\bx +\by) \subseteq  \Rsupp(\bx) + \Rsupp(\by)$.\\
$(3)$ Let $\bx, \by \in L^n$. Then $\wt_R(\bx+\by ) \leq \wt_R(\bx) + \wt_R(\by)$.\\
$(4)$ If $\bg_1, \ldots ,\bg_k$ generate $C$ as an $L$-linear space, then $\Rsupp(C)$
is the $K$-linear sum of the $\Rsupp(\bg_i)$, $i=1, \ldots ,k$.
\end{proposition}

\begin{proof}
$(1)$  Let $B=\{ \alpha_1, \ldots ,\alpha _m \}$ be a basis of $L$ as a vector space over $K$.
Let $\alpha \in L^*$. Then $B'=\{ \alpha \alpha_1, \ldots ,\alpha \alpha _m \} $
is another basis of $L$ as a vector space over $K$.
Now $\Rsupp(\bx )$ is the row space of $M_B(\bx )$.
But this row space does not depend on the chosen basis.
Hence $\Rsupp(\alpha \bx )$ is the row space of $M_{B'}(\alpha \bx )$
and this row space is equal to $\Rsupp(\bx )$, since $M_{B'}(\alpha \bx ) = M_{B}(\bx )$.\\
$(2)$ The matrix $M_B(\bx )$ has entries $x_{ij}$ with $x_j = \sum_{i=1}^m x_{ij}\alpha_i$, and
$M_B(\by )$ has entries $y_{ij}$ with $y_j = \sum_{i=1}^m y_{ij}\alpha_i$.
So $M_B(\bx +\by )$ has entries $x_{ij}+y_{ij}$, since $x_j +y_j= \sum_{i=1}^m (x_{ij}+y_{ij})\alpha_i$.
Now  $\Rsupp(\bx)$,  $\Rsupp(\by)$ and $\Rsupp(\bx +\by)$ are equal to the row spaces of
$M_B(\bx)$, $M_B(\by)$ and $M_B(\bx +\by)$, respectively.
Therefore $\Rsupp(\bx +\by) \subseteq  \Rsupp(\bx) + \Rsupp(\by)$, since $M_B(\bx +\by )=M_B(\bx)+M_B(\by)$.\\
$(3)$ This is a direct consequence of $(2)$, since $\wt_R(\bx) =\dim \Rsupp(\bx)$
and $\dim (I+J) \leq \dim (I) + \dim (J)$ for subspaces $I$ and $J$ of $K^n$.\\
$(4)$  Let $\bg_1, \ldots ,\bg_k$ generate $C$ as an $L$-linear space.
Then $\Rsupp(C)$ contains  the $K$-linear sum of the $\Rsupp(\bg_i)$, by definition.
Conversely, let $\bc \in C$. Then $\bc = \sum _{i=1}^k \lambda _i \bg_i$ is in the $K$-linear sum
of the $\Rsupp(\bg_i)$, by applying $(1)$ and $(2)$ repeatedly.
Therefore $\Rsupp(C)$ is contained in   the $K$-linear sum of the $\Rsupp(\bg_i)$.
\end{proof}

\begin{definition}
Let $K=\fq$ and $L=\fqm$.
Let $C$ be an $L$-linear code. Then for every $r=0,\ldots,k$ the \emph{generalized rank weight enumerator} is defined by
\[ W_C^{R,r}(X,Y)=\sum_{w=0}^n A_w^{R,r}X^{n-w}Y^w, \]
where $A_w^{R,r}$ is the number of subcodes of $C$ of dimension $r$ and rank weight $w$.
This is well-defined, since $L^n$ is finite.
\end{definition}

Just like with the weight enumerator and the minimum distance, a special case of interest
is the first nonzero coefficient of these polynomials.

\begin{definition}\label{d-grwJP}
Let $C$ be an $L$-linear code.
Then $d_{R,r}(C)$, the $r$-th {\em generalized rank weight} of the code $C$ is
the minimal rank support weight of a subcode $D$ of $C$ of dimension $r$. That is:
\[ d_{R,r}(C)=\min_{\substack{D \subseteq C \\ \dim(D) =r}}\wt_R(D). \]
\end{definition}

The above is not the only proposed definition of the generalized rank weights. The first proposal of a definition of the $r$-th generalized rank weight was given by Kurihara-Matsumoto-Uyematsu \cite{kurihara:2012,kurihara:2013}.
An alternative was given by Oggier-Sboui \cite{oggier:2012} and Ducoat \cite{ducoat:2013}. Both definitions were motivated by applications. Before we discuss these definitions, we develop some more theory about rank metric codes.

\section{Some codes related to $C$}

With respect to the Hamming distance and a $k$-dimensional $\fq $-linear code $C$,
the support of $C$ is defined by $\supp(C)= \{j\ |\ c_j\not=0 \mbox{ for some } \bc \in C \}$.
The subcode $C(J)$ is defined in \cite{katsman:1987} and \cite[Definition 5.1]{jurrius:2013}
for a subset $J$ of $\{1, \ldots ,n\}$ with complement $J^c$ by:
$$
C(J)=\{ \ \bc \in C \ | \ \supp(\bc) \subseteq J^c \ \}.
$$
Define $C(j)=C(\{ j \})$ for $j \in \{1, \ldots ,n\}$. Let $J=\supp(C)$, then $C(j)$ has codimension $1$ in $C$ for all $j \in J$.
In fact, $C(j)$ is the code $C$ punctured in position $j$, but without the removing of the zeros in this position. \\

We claim that if $n<q$, then there exists a $\bc \in C$ such that $\supp(\bc ) =\supp(C)$.
Suppose $n<q$. Then $|J|<q$ and
$$
\bigcup_{j\in J} C(j) \not=C,
$$
since
$$
\left|\bigcup_{j\in J} C(j)\right| \leq |J| q^{k-1} < q^k = |C|.
$$
Let $\bc \in C\setminus \bigcup_{j\in J} C(j)$. Then $\supp (\bc ) \subseteq \supp(C) =J$. Furthermore $\bc \not\in C(j)$ for all $j\in J$.
So $c_j\not=0$ for all $j\in J$. Hence $J= \supp (\bc ) = \supp(C)$.\\

We will now translate this statement to rank metric codes. Notice that the statements ``$I\subseteq J^c$" and ``$I\cap J =\emptyset$" are equivalent for two subsets $I$ and $J$ of $\{1, \ldots ,n\}$.
These statements would translate into ``$I \subset J^\perp $" and ``$I\cap J =\{ 0 \}$", respectively, for subspaces $I$ and $J$ of $K^n$, but these are not equivalent.
For the definition of $C(J)$ in the context of the rank metric we give the following analogous definition as given in \cite[Definition 2]{jurrius:2014}.

\begin{definition}\label{dJ2}
Let $L$ be a finite field extension of the field $K$. Let $C$ be an $L$-linear code.
For a $K$-linear  subspace $J$ of $K ^n$ we define:
\[ C(J) = \{ \ \bc \in C \ | \ \Rsupp(\bc ) \subseteq J^\perp \ \} \]
\end{definition}

From this definition it is clear that $C(J)$ is a $K$-linear subspace of $C$, but in fact it is also an $L$-linear subspace.

\begin{lemma}\label{lJ2}
Let $C$ be an $L$-linear code of length $n$ and let $J$ be a $K$-linear  subspace of $K ^n$.
Then $\bc \in C(J)$ if and only if  $\bc \cdot \by =0 $ for all $ \by \in J $.
Furthermore $C(J)$ is an $L$-linear subspace of $C$.
\end{lemma}
\begin{proof}
The following statements are equivalent:
$$
\begin{array}{c}
\bc \in C(J)\\
\sum_{j=1}^nc_{ij}y_j =0 \mbox{ for all } \by \in J \mbox{ and } i=1, \ldots ,m\\
\sum_{i=1}^m(\sum_{j=1}^nc_{ij}y_j) \alpha_i=0 \mbox{  for all } \by \in J\\
\sum_{j=1}^n(\sum_{i=1}^mc_{ij} \alpha_i) y_j=0 \mbox{ for all } \by \in J\\
\sum_{j=1}^n c_j y_j=0 \mbox{ for all }\by \in J\\
\bc \cdot \by =0\mbox{ for all } \by \in J\\
\end{array}
$$
Hence $C(J) = \{  \bc \in C\ |\  \bc \cdot \by =0 \mbox{ for all } \by \in J  \} $.
From this description it follows straightforwardly that $C(J)$ is an $L$-linear subspace of $C$.
\end{proof}

\begin{corollary}\label{cJ2}
Let $C$ be an $L$-linear code of length $n$. Let $J$ be a $K$-linear  subspace of $K^n$.
Then  $\dim_L (C(J)) \geq \dim _L(C) - \dim _K(J)$.
\end{corollary}

\begin{proof}
This follows directly from Lemma \ref{lJ2}.
\end{proof}

\begin{remark}
Let $I$ be a  $K$-linear subspace of $K ^n$ and $J= \Rsupp (C)$.
Then $C=C(I)$ if and only if $I \subseteq J^\perp$, since the following statements are equivalent:
$$
\begin{array}{c}
C= C(I)\\
\Rsupp(\bc ) \subseteq I^\perp \mbox{ for all } \bc \in C\\
\Rsupp(C) \subseteq I^\perp\\
J \subseteq I^\perp\\
I \subseteq J^\perp\\
\end{array}
$$
Hence, it is not necessarily the case that $\dim_L (C(I)) \geq \dim _L(C) -1$ for all
one dimensional subspaces $I$ of $J$, since we might have that $J \subseteq J^\perp$.
\end{remark}

\begin{proposition}\label{p-RsuppC}
Let $L=\fqm$ and $K=\fq$. Let $C$ be an $L$-linear code.
If $m \geq n$, then there exists a $\bc \in C$ such that
$$
\Rsupp (\bc) = \Rsupp (C).
$$
\end{proposition}

\begin{proof}
First, note that the inclusion $\Rsupp (\bc) \subseteq \Rsupp (C)$ holds for all $\bc \in C$.
Let $k =\dim_L (C)$ and $J = \Rsupp(C)$. Then we have
$$
\left|\bigcup_{\substack{\dim(I)=1 \\ C(I)\not=C}} C(I)\right| \leq \frac{q^n-1}{q-1}q^{m(k-1)} < q^{mk} = |C|.
$$
In the first inequality, we use that the number of one dimensional subspaces of $K^n$ is $\frac{q^n-1}{q-1}$,
since $K =\fq$.
In the second inequality, we use that $ n\leq m$ so $\frac{q^n-1}{q-1}<q^m$. It follows that
$$
\bigcup_{\substack{\dim(I)=1 \\ C(I)\not=C}} C(I) \not=C.
$$
Let $\bc \in C\setminus \bigcup_{\dim(I)=1, C(I)\not=C} C(I)$. Then by definition $\Rsupp (\bc ) \subseteq J$.
Now suppose $\Rsupp (\bc ) \not=J$. Then $\Rsupp(\bc)$ is contained in a codimension $1$ subspace of $J$, hence there exists a codimension $1$ subspace $H$ of $K^n$ such that
$\Rsupp (\bc ) \subseteq H$ and $J\cap H\not= J$. \\
Let $I=H^\perp$. Then $I$ is a $1$-dimensional subspace of $K^n$ with $\Rsupp (\bc ) \subseteq H=I^\perp$, so $\bc \in C(I)$.
Now $C(I)=C$ by the choice of $\bc$,  hence $\Rsupp (\bx ) \subseteq I^\perp=H$ for all $\bx \in C$. Therefore $J=\Rsupp (C) \subseteq H$.
This is a contradiction, since $J\cap H\not= J$.
So $\Rsupp (\bc) = \Rsupp (C)$.
\end{proof}

\begin{remark}
We could not proof nor disproof this proposition for arbitrary field $K$ and a finite extension $L$ of degree $m\geq n$. The following example gives a counterexample in the case $m<n$.
\end{remark}

\begin{example}\label{ex-m<n}
Let $K=\ff _2$, $m=3$ and $L=\ff_8$. Let $\alpha \in L$ with $\alpha ^3=1+\alpha $.
Let $C$ be the $L$-linear code in $L^4$ generated by $\ba =(1,\alpha, \alpha ^2, \alpha ^3)$
and $\bb =(1,\alpha, \alpha ^2, \alpha ^4)$.
Then $\Rsupp (C)=K^4\not=\Rsupp (\bc)$ for all $\bc \in C$, since $\dim \Rsupp (\bc)\leq m=3$.
\end{example}

\begin{proposition}\label{p-Rsupp-x+y}
Let $L=\fqm$ and $K=\fq$.
Let $\bx, \by \in L^n$. Let $I=\Rsupp(\bx)$ and $J=\Rsupp(\by)$.
If $m\geq n$, then there are constants $\alpha, \beta \in L$ such that $\Rsupp(\alpha \bx+\beta\by)=I+J$.
\end{proposition}

\begin{proof}
This is a consequence of Proposition \ref{p-RsuppC}.
\end{proof}

\section{Galois closure and trace}

Before we can give the various definitions of the generalized rank weights, we introduce the framework in which we study them.

\begin{definition}
Let $L/K$ be a Galois extension.
Let $C\subseteq L^n$ be an $L$-linear subspace.
The {\em trace map} $\Tr: L^n \rightarrow K^n$ is the component-wise extension of the trace map $\Tr: L \rightarrow K$.
The {\em restriction} of $C$ is defined by $C|_K = C \cap K^n$.
The \emph{Galois closure} $C^*$ of $C$ is the smallest subspace of $L^n$ that contains $C$ and that is closed under the component-wise action of the Galois group of $L/K$. A subspace is called \emph{Galois closed} if and only if it is equal to its own Galois closure. \\
If $C$ is a $K$-linear subspace, then we define the \emph{extension codes} $C\otimes L$ as the subspace of $L^n$ formed by taking all $L$-linear combinations of words of $C$.
\end{definition}

We can summarize the above relations in the following diagram:
\[ \xymatrix{
C \ar[r]^{\text{Gal}(L/K)} \ar[dr]^{\Tr(\mathbf{c})} & C^* \\
C|_K \ar@{^{(}->}[u] \ar@{^{(}->}[r] & \Tr(C) \ar[u]_{\otimes L}
} \]

The codes on the top row are over $L$, those on the lower row over $K$. The following Theorem is based on Giorgetti-Previtali \cite{giorgetti:2010}:
\begin{theorem}
Let $L/K$ be a Galois extension. Let $C$ be an $L$-linear code. Then the following statements are equivalent:
\begin{itemize}
\item $C$ is Galois closed: $C=C^*$.
\item $C$ is the extension of its restriction: $C=(C|_K) \otimes L$.
\item $C$ has a basis over $K^n$.
\item The trace of $C$ is equal to its restriction: $\Tr (C) = C|_K$.
\end{itemize}
\end{theorem}

Interpreted in the diagram above, this theorem says that the two codes at the top are the same if and only if the two codes on the bottom are the same. This leads to an interesting observation about rank metric codes.

\begin{theorem}
Let $\bc\in C$. Then the rows of the matrix $M(\bc)$ are elements of the trace code $\Tr(C)$ and $\Rsupp(C)=\Tr(C)$.
\end{theorem}
\begin{proof}
The extension $L/K$ is Galois, so in particular separable. The product defined by $\langle x,y\rangle:=\Tr(xy)$
is a $K$-bilinear non-degenerate inner product on $L$, see \cite[VII, \S 5, Theorem 9]{lang:1965}.
For the given $K$-linear basis  $\alpha _1, \ldots , \alpha _m$ of $L$
there exists a $K$-linear basis  $\alpha '_1, \ldots , \alpha '_m$ of $L$ such that
$\langle\alpha _i , \alpha '_j\rangle=\delta_{ij}$ is the Kronecker delta function, see \cite[VII, \S 5, Corollary 2]{lang:1965}.
So  $M(\bc)$ has entries $c_{ij}$ in $K$ such that $c_j = \sum_{i=1}^m c_{ij} \alpha _i$,  $\bc _i =(c_{i1}, \ldots , c_{in})$ is the $i$-th row of $M(\bc)$ and $\bc = \sum_{i=1}^m \bc_i\alpha _i$.
The trace map is $K$-linear . Hence
$$
\Tr (\alpha '_j\bc )= \sum_{i=1}^m \bc_i \Tr(\alpha _i \alpha '_j)= \sum_{i=1}^m \bc_i \delta_{ij} = \bc _j.
$$
Hence the rows of the matrix $M(\bc)$ are elements of the trace code $\Tr(C)$.
So $\Rsupp(C)\subseteq\Tr(C)$, since $\Rsupp(C)$ is generated by the rows of $M(\bc)$. \\
Now we consider the converse inclusion. There  exists a $\beta_1 \in L $ such that $\Tr (\beta_1 )=1$. Let $\beta _2, \ldots , \beta _m$
be a $K$-linear basis of the kernel of $\Tr$. Then $\beta _1, \ldots , \beta _m$ is a $K$-linear basis of $L$
such that $\Tr (\beta _i)$ is one if $i=1$ and is zero otherwise.
Without loss of generality we may assume that the matrix $M(\bc)$ is obtained with respect to this basis.
So  now $M(\bc)$ has entries $c'_{ij}$ with $c_j = \sum_{i=1}^m c'_{ij} \beta _i$,
$\bc' _i =(c'_{i1}, \ldots , c'_{in})$ is the $i$-th row of $M(\bc)$
and $\bc = \sum_{i=1}^m \bc'_i\beta _i$. Hence $\Tr (\bc )= \sum_{i=1}^m \bc'_i \Tr(\beta _i)= \bc'_1 \in \Rsupp(C)$.
Therefore $\Tr(C)\subseteq \Rsupp(C)$.
\end{proof}

\begin{corollary}
Let $D$ be a subcode of the $L$-linear code $C$. Then $\Rsupp(D)=\Tr(D)$ and thus
\[ d_{R,r}(C)=\min_{\substack{D \subseteq C \\ \dim(D) =r}}\wt_R(D)=\min_{\substack{D \subseteq C \\ \dim(D) =r}}\dim \Tr(C)=\min_{\substack{D \subseteq C \\ \dim(D) =r}}\dim D^* \]
\end{corollary}

\section{Equivalent definitions}

We will now discuss previous definitions of the generalized Hamming weights and to what extend they are consistent with Definition \ref{d-grwJP}. The definition of Oggier-Sboui in \cite{oggier:2012} is, in our notation, as follows:

\begin{definition}\label{d-grwOS}
Consider the field extension $\mathbb{F}_{q^m}/\mathbb{F}_q$. Let $C$ be an $\mathbb{F}_{q^m}$-linear code and let $m\geq n$. Then the $r$-th generalized rank weight is defined as
\[ \min _{\substack{D \subseteq C \\ \dim(D) =r}} \max_{\bd \in D} \wt_R(\bd). \]
\end{definition}

\begin{theorem}
Let $L=\fqm$ and $K=\fq$. Let $C$ be an $L$-linear code with $m\geq n$. Then Definitions \ref{d-grwOS} and \ref{d-grwJP} give the same values, that is,
\[ d_{R,r}(C)=\min_{\substack{D \subseteq C \\ \dim(D) =r}}\wt_R(D) = \min _{\substack{D \subseteq C \\ \dim(D) =r}} \max_{\bd \in D} \wt_R(\bd). \]
\end{theorem}

\begin{proof}
By Proposition \ref{p-RsuppC}, every subcode $D$ contains a word of maximal rank weight.
\end{proof}

Kurihara-Matsumoto-Uyematsu \cite{kurihara:2012,kurihara:2013} define the \emph{relative generalized rank weights}, that induce the following definition of the generalized rank weights:

\begin{definition}\label{d-grwKMU}
Consider the field extension $\mathbb{F}_{q^m}/\mathbb{F}_q$. Let $C$ be an $\mathbb{F}_{q^m}$-linear code. Then the $r$-th generalized rank weight is defined as
\[ \min_{\substack{V\subseteq L^n, V=V^* \\ \dim(C\cap V)\geq r}}\dim V. \]
\end{definition}

Both of Definitions \ref{d-grwOS} and \ref{d-grwKMU} have an obvious extension to rank metric codes over the field extension $L/K$. Where possible, we will show the equivalence between the definitions in as much generality as possible. \\
Ducoat \cite{ducoat:2013} proved the following for $m\geq n$:
\[ \min _{\substack{D \subseteq C \\ \dim(D) =r}} \max_{\bd \in D^*} \wt_R(\bd)=\min_{\substack{V\subseteq L^n, V=V* \\ \dim(C\cap V)\geq r}}\dim V. \]
The left hand side is almost Definition \ref{d-grwOS}, but with $D^*$ instead of $D$ in the maximum.

The following proof is largely inspired by Ducoat; note that it works more general over $L$ instead of $\mathbb{F}_{q^m}$:

\begin{theorem}Let $L$ be a Galois extension of $K$.
Let $C$ be an $L$-linear code of dimension $k$. Let $r$ be an integer such that $0\leq r \leq k$.
Then Definitions \ref{d-grwKMU} and \ref{d-grwJP} give the same values, that is,
\[ \min_{\substack{V\subseteq L^n, V=V^* \\ \dim(C\cap V)\geq r}}\dim V = \min_{\substack{D\subseteq C \\ \dim D=r}}\dim D^*. \]
\end{theorem}
\begin{proof}
Let $V$ a Galois closed subspace of $L^n$ such that $\dim(C\cap V)\geq r$. (Such a $V$ always exists, since $C^*$ is such a subspace and $r\leq k$.) Let $D\subseteq C\cap V$ with $\dim D=r$. Since $V$ is Galois closed and $D\subseteq V$, we have $D^*\subseteq V$. So $D^*$ is a subspace of smaller dimension than $V^*$ with $\dim(C\cap D^*)\geq r$. On the other hand, for all $D\subseteq C$ with $\dim D=r$, we have that $\dim(C\cap D^*)\geq r$. So the formula follows.
\end{proof}

We will now continue to prove the equivalence of Definition \ref{d-grwJP} and the variation of Definition \ref{d-grwOS} that was used by Ducoat \cite{ducoat:2013}. This requires some propositions. For the rest of this section, let $L/K$ be a cyclic Galois extension of degree $m$ --- we make this assumption so we can use the results in \cite{augot:2013}. Enumerate the elements of the Galois group as $\theta_i$ and let $\bx^{[i]}=\bx^{\theta_i}$ be the (component-wise) action of the Galois group.

\begin{proposition}
Let $L$ be a cyclic Galois extension of $K$.
For any $\bx \in L^n$ we have
\[ \langle\bx\rangle^*=\Rsupp(\bx)\otimes L. \]
\end{proposition}
\begin{proof}
First we note that $\langle\bx\rangle^*=\langle\bx,\bx^{[1]},\ldots,\bx^{[m-1]}\rangle$. We prove equality by proving two inclusions. Let $(\alpha_1,\ldots,\alpha_m)$ a basis of $L$ of $K$ and write $\bx=\alpha_1\bx_1+\ldots+\alpha_m\bx_m$. This means the $\bx_i$ are the rows of $M(\bx)$. We have that
\[ \bx^{[i]}=(\alpha_1\bx_1+\ldots+\alpha_m\bx_m)^{[i]}=\alpha_1^{[i]}\bx_1+\ldots+\alpha_m^{[i]}\bx_m \]
and $\alpha_j^{[i]}\in L$, so $\bx^{[i]}\in\langle\bx_1,\ldots,\bx_m\rangle\otimes L$ and thus
$\langle \bx \rangle^* \subseteq \Rsupp(\bx) \otimes L$. \\
For the reverse inclusion we may assume without loss of generality after a permutation of coordinates that the first $l$ columns of $M(\bx)$ are independent, where $l$ is the rank of $M(\bx)$, and moreover that $M(\bx)$ is in reduced row echelon form. Hence $\bx=\alpha_1\bx_1+\ldots+\alpha_l\bx_l$ and $\bx_j = {\bf 0}$ for all $l<j\leq m$. Then
\[ \bx^{[i]} = \alpha_1^{[i]}\bx_1+\ldots+\alpha_l^{[i]}\bx_l  \ \mbox { for all } \  i=0, \ldots ,l-1.\]
(This holds in fact for $i=0,\ldots,m-1$, but we only need it up to $i=l-1$.)
The ``Vandermonde'' matrix with entries $\alpha _j^{[i-1]}$, $1 \leq i,j \leq l$ is invertible by
\cite[Theorem 3]{augot:2013}, since the $\alpha_1, \ldots,\alpha_l$ are  independent over $K $.
So $\bx_j\in \langle \bx^{[0]},\ldots \bx^{[l-1]} \rangle \subseteq \langle\bx\rangle^*$.
Therefore $\Rsupp(\bx)\otimes L \subseteq \langle \bx \rangle^*$ and we conclude that
$\langle \bx \rangle^* = \Rsupp(\bx) \otimes L$.
\end{proof}

From this Proposition, the following Lemmas follows directly. They are a generalization of I.1 and II.2 of Ducoat \cite{ducoat:2013}.

\begin{lemma}
Let $L$ be a cyclic Galois extension of $K$.
For any $\bx \in L^n$ we have
\[ \dim(\langle \bx \rangle^*)=\rk(M(\bx)). \]
\end{lemma}

\begin{lemma}
Let $L$ be a cyclic Galois extension of $K$.
For all Galois closed sets $V$ of dimension $n\leq m$ there is an $\bx\in V$ such that $V=\langle\bx\rangle^*$.
\end{lemma}
\begin{proof}
Pick a basis for $V$ and let these be rows of $M(\bx)$. Add extra zero rows.
\end{proof}

The first Lemma originally had the assumption $m\geq n$, but this can be dropped. In the second Lemma we can not do that: the number of basis vectors for $\langle\bx\rangle^*$ is equal to the extension degree $m$, so we can never have $V=\langle\bx\rangle^*$ if $\dim V>m$. (See also Example \ref{ex-m<n}.) \\
We now show the equivalence between Definitions \ref{d-grwOS} and the variation of \ref{d-grwJP} as used before.

\begin{theorem}
Let $L$ be a cyclic Galois extension of $K$ of degree $m$.
Let $C$ be an $L$-linear code in $L^n$ with $m\geq n$. Then
\[ \max_{\bd\in D^*} \rk(M(\bd))=\dim D^* \]
\end{theorem}
\begin{proof}
Because $D^*$ is Galois closed, there is a $\bd\in D^*$ such that $D^*=\langle\bd\rangle^*$. We also have $\dim(\langle\bd\rangle^*)=\rk(M(\bd))$, so $\rk(M(\bd))=\dim D^*$ and indeed
\[ \max_{\bd\in D^*} \rk(M(\bd))=\dim D^*. \]
\end{proof}

\section{Degenerate codes}

As a small application of the generalized rank weights, we discuss the concept of \emph{degenerate} codes. A code $C$ is called degenerate with respect to the Hamming metric in case there is a coordinate such that all codewords are zero at that position. This is equivalent with saying that the Hamming minimum distance of $C^\perp$ is one.

\begin{definition}
The code $C$ is called degenerate with respect to the rank metric if $d_R(C^\perp)=1$.
\end{definition}

The dual code here is defined as the orthogonal subspace of $C$ in $L^n$. We give some characteristics of (non)degenerate codes.

\begin{proposition}\label{p-deg}
The code $C$ is degenerate with respect to the rank metric if and only
$C$ is rank equivalent with a code such that all its codewords are zero at the last position.
\end{proposition}
\begin{proof}
Suppose that $d_R(C^\perp)=1$.
Then there is a $\bh \in C^\perp$ such that  $M_B(\bh )$ has rank $1$.
After dividing by the first element of the basis $B$ we may assume that
the first element of $B$ is equal to $1$.
After a change of the basis $B$ we may assume that
only the first row of $M_B(\bh )$ has nonzero entries.
Hence $\bh \in K^n$, since the first element of $B$ is equal to $1$.
We can take this $\bh $ as the first row of the parity check matrix of $C$.
After a coordinate change of $L^n$ with entries in $K$, that is a rank isometry, we may assume
that we get a rank equivalent code with $\bh =(0,\ldots,0,1)$.
Therefore $C$ is rank equivalent with a code such that all its codewords are zero at the
last position. \\
Clearly the converse also holds.
\end{proof}

\begin{corollary}\label{cor1-deg}
Let $C$ be an $L$-linear code of length $n$ and dimension $k$.
Then $C$ is nondegenerate with respect to the rank metric if and only if $\Rsupp(C)=K^n$ if and only if $d_{R,k}(C)=n$.
\end{corollary}
\begin{proof}
Suppose that $\Rsupp(C)\not=K^n$. Then there is a coordinate transformation with entries in $K$ such
that $\Rsupp(C) \subseteq K^{n-1}\times\{ 0 \}$. Hence $C$ is degenerate with respect to the rank metric
by Proposition \ref{p-deg}. Clearly the converse also holds.\\
The last equivalence follows from the definition of $d_{R,k}(C)$.
\end{proof}

An alternative proof of this theorem uses the following result on duality by Ducoat \cite[Theorem I.3]{ducoat:2013}:

\begin{theorem}
Let $C$ be an $L$-linear code with dual $C^\perp$. Then the generalized rank weights of $C$ and $C^\perp$ are related as follows:
\[ \{d_{r,R}(C)\ |\ 1\leq r\leq k\} = \{1,\ldots,n\}\setminus\{n+1-d_{r,R}(C^\perp)\ |\ 1\leq r\leq n-k\}. \]
\end{theorem}

\begin{proof}[Alternative proof of Corollary \ref{cor1-deg}]
If $d$ is a generalized rank weight of $C$, then $n+1-d$ is not a generalized rak weight of $C^\perp$. So if $d_{R,1}(C^\perp)=1$, then $n+1-1=n$ can not be a generalized rank weight of $C$, so $d_{R,k}(C)<n$.
\end{proof}

\begin{corollary}\label{cor2-deg}
Let $L$ be an extension of $K$ of degree $m$.
Let $C$ be an $L$-linear code of length $n$ and dimension $k$ that is nondegenerate with respect to the rank metric.
Then $km \geq n$.
\end{corollary}

\begin{proof}
Let $\bg_1, \ldots ,\bg_k$ be a basis of $C$. Then $\Rsupp(C)$ is generated by the spaces
$\Rsupp(\bg_i)$ for $i=1, \ldots ,k$  by Proposition \ref{p-Rsupp}. The dimension of $\Rsupp(\bg_i)$ is at most $m$.
Hence the dimension of $\Rsupp(C)$ is at most $km$ by Proposition \ref{p-Rsupp}.
If $C$ is nondegenerate, then $\Rsupp(C)=K^n$ by Corollary \ref{cor1-deg}. Therefore $km \geq n$.
\end{proof}

\section{Conclusion}

This paper introduced a new definition for the generalized rank weights. This definition is induced by the study of the generalized rank weight enumerator \cite{jurrius:2014}. We investigated the relation between the rank support and the trace code. The fact that the proposed definition of the generalized rank weight enumerator induces a definition of the generalized rank weights that is equivalent to a known definition, supports the definition of the generalized rank weight enumerator. \\

We showed that our definition is equivalent to that of Kurihara-Matsumoto-Uyematsu \cite{kurihara:2012,kurihara:2013}, even if we extend this definition to codes over arbitrary  fields. We also show that our definition of the generalized rank weight is equivalent to the definition of Oggier-Sboui \cite{oggier:2012}: here we have to assume $m\geq n$ and that the Galois group of $L/K$ is cyclic. \\
The fact that the proposed definition of the generalized rank weight enumerator induces a definition of the generalized rank weights that is equivalent to a known definition, supports the definition of the generalized rank weight enumerator. \\
We believe that the fact that we were able to prove equivalence to the definition of Kurihara-Matsumoto-Uyematsu but only partial equivalence to the definition of Oggier-Sboui, supports the definition of the generalized rank weights of Kurihara-Matsumoto-Uyematsu in favor of the definition of Oggier-Sboui. A similar conclusion can be drawn from the work of Ducoat \cite{ducoat:2013}, who proves duality relations for the generalized rank weights. \\

As an application of the notion of generalized rank weights, we discussed codes that are degenerate with respect to the rank metric. \\

As further work, we propose to study our definition for Delsarte codes
\cite{delsarte:1978,ravagnani:2014}: these are rank metric codes that are linear subspaces of $K^{m\times n}$ that do not necessarily come from a code over $L$. Our definition, contrary to that of Kurihara-Matsumoto-Uyematsu, is directly applicable to such codes.

\section{Acknowledgement}
The first author is partially supported by grant GOA62 of the Vrije Universiteit Brussel, Belgium.
The second author is indebted to the Aalborg University for a stay during which this research could be finished and
that was supported by the Danish Council for Independent Research (Grant No. DFF 4002-00367).

\bibliographystyle{plain}
\bibliography{rkwtenum}

\end{document}